\documentclass[a4paper,UKenglish]{lipics-v2016}

\usepackage{amssymb,amsmath,amsthm}
\usepackage{float}
\usepackage{changepage}
\usepackage{enumerate} % lists with (i), (ii), (iii), ...
\usepackage{tikz}
\usepackage{thmtools} % autoref
\usepackage{hyperref, cleveref} %autoref
\usepackage{breqn} %multiple line equations
\newcommand{\ra}{\rightarrow}

\theoremstyle{plain}
\newtheorem{claim}{Claim}[theorem]

\newtheorem{observation}[theorem]{Observation}

\newcommand{\data}{\textbf{Input: }}
\newcommand{\quest}{\textbf{Output: }}
\newcommand{\res}{\textbf{Result: }}
\newcommand{\opt}{\textbf{Optimization: }}
\newenvironment{oproblem}[4]{\begin{figure}[H]\begin{adjustwidth}{0.5cm}{}{\sc#1}\\ \data #2\\ \res #3\\ \opt #4}{\end{adjustwidth}\end{figure}\vspace{-1em}}

\newcommand{\x}{X}

\newcommand{\Vx}{V_\x}
\newcommand{\Ax}{A_\x}
\newcommand{\p}{U}
\newcommand{\B}{B}
\newcommand{\M}{M}
\newcommand{\Ap}{A'}
\newcommand{\Apb}{\overline{\Ap}}
\newcommand{\Up}{U'}
\newcommand{\Upb}{\overline{\Up}}
\newcommand{\Uf}{\p[F]}

\newcommand{\ord}{\sigma}
\newcommand{\ba}{b}

\newcommand{\lo}{<_\ord}
\newcommand{\Gf}{G'}
\newcommand{\Vf}{V'}
\newcommand{\Af}{A'}
\newcommand{\Pip}{\Pi'}

\newcommand{\T}{{\cal T}}
\newcommand{\V}{\sigma}

\renewcommand{\t}{\delta}

\renewcommand{\S}{\Delta}
\newcommand{\SC}{O}
\newcommand{\Tv}{T_{v}}
\newcommand{\Tc}{T_{c}}
\newcommand{\Vv}{\V}
\newcommand{\Vc}{\V}
\newcommand{\Vcv}{\V^{cv}}
\newcommand{\SVVV}{\S_{V,V,V}}
\newcommand{\SVVC}{\S_{V,V,C}}
\newcommand{\SVCC}{\S_{V,C,C}}
\newcommand{\SCCC}{\S_{C,C,C}}
\newcommand{\SdVC}{\S_{2V,C}}
\newcommand{\SVdC}{\S_{V,2C}}
\newcommand{\StV}{\S_{3V}}
\newcommand{\StC}{\S_{3C}}

\newcommand{\So}{\S}
\newcommand{\SoVVV}{\So_{V,V,V}}
\newcommand{\SoVVC}{\So_{V,V,C}}
\newcommand{\SoVCC}{\So_{V,C,C}}
\newcommand{\SoCCC}{\So_{C,C,C}}
\newcommand{\SodVC}{\So_{2V,C}}
\newcommand{\SoVdC}{\So_{V,2C}}

\newcommand{\SotV}{\So_{3V}}
\newcommand{\SotC}{\So_{3C}}
\newcommand{\Sv}{\S^{v}}

\newcommand{\Sc}{\S^{c}}

\newcommand{\CVV}{\SC_{V^*}}
\newcommand{\CV}{\SC_{V}}
\newcommand{\CCC}{\SC_{C^*}}
\newcommand{\CC}{\SC_{C}}
\newcommand{\CVC}{\SC_{V^*,C^*}}

\newcommand{\maxtst}{{\sc MaxTST}}
\newcommand{\maxcst}{{\sc MaxCST}}

\newcommand{\maxtt}{{\sc MaxTT}}
\newcommand{\kmaxtt}{$k$-\maxtt}

\newcommand{\maxct}{{\sc MaxCT}}
\newcommand{\kmaxct}{$k$-\maxct}

\newcommand{\maxtttight}{{\sc Tight-MaxTT}}
\newcommand{\maxcttight}{{\sc Tight-MaxCT}}

\newcommand{\tsatt}{{\sc 3-SAT(3)}}

\newcommand{\NP}{{\sf NP}}

\newcommand{\FPT}{{\sf FPT}}

\newlength{\mylfA}
\newlength{\myhfA}

\newcommand{\fA}[1]{%
\settowidth{\mylfA}{#1}
\def\sizefB{\the\mylfA}
\settoheight{\myhfA}{#1}
\def\heightfB{\the\myhfA}
\hspace{-0.35em}
\begin{tikzpicture}[baseline=(O.base)] 
\draw[ -> ](0,\heightfB)  -- ( \sizefB, \heightfB);
\node (O) at (\sizefB/2,\heightfB/7-0.01ex) {$#1$};
\end{tikzpicture}%
\hspace{-0.35em}
}

\newcommand{\bA}[1]{%
\settowidth{\mylfA}{#1}
\def\sizefB{\the\mylfA}
\settoheight{\myhfA}{#1}
\def\heightfB{\the\myhfA}
\hspace{-0.35em}
\begin{tikzpicture}[baseline=(O.base)]
\draw[ <- ](0,\heightfB)  -- ( \sizefB, \heightfB);
\node (O) at (\sizefB/2,\heightfB/7-0.01ex) {$#1$};
\end{tikzpicture}%
\hspace{-0.35em}
}

\title{(Arc-disjoint) cycle packing in tournament: classical and parameterized complexity}%\footnote{This work was partially supported by someone.}}
\author[1]{St\'ephane Bessy}
\author[1]{Marin Bougeret}
\author[1]{Jocelyn Thiebaut}
\affil[1]{Universit\'e de Montpellier, LIRMM, CNRS, Montpellier, France\\
  \texttt{\{bessy, bougeret, thiebaut\}@lirmm.fr}}
\authorrunning{S.\,Bessy,  M.\,Bougeret and J.\,Thiebaut} %mandatory. First: Use abbreviated first/middle names. Second (only in severe cases): Use first author plus 'et. al.'

\Copyright{St\'ephane Bessy, Marin Bougeret and Jocelyn Thiebaut}%mandatory, please use full first names. LIPIcs license is "CC-BY";  http://creativecommons.org/licenses/by/3.0/

% mandatory: Please choose ACM 1998 classifications from http://www.acm.org/about/class/ccs98-html . E.g., cite as "F.1.1 Models of Computation". 
% mandatory: Please provide 1-5 keywords
\begin{document}

\maketitle

\begin{abstract}
Given a tournament ${\cal T}$, the problem \maxct~consists of finding a maximum (arc-disjoint) cycle packing of $\T$. In the same way, \maxtt~corresponds to the specific case where the collection of cycles are triangles (i.e. directed 3-cycles). Although \maxct~can be seen as the LP dual of minimum feedback arc set in tournaments which have been widely studied, surprisingly no algorithmic results seem to exist concerning the former. 

In this paper, we prove the \NP-hardness of both \maxct~and \maxtt. We also show that deciding if a tournament has a cycle packing and a feedback arc set with the same size is an \NP-complete problem. In light of this, we show that \maxtt~admits a vertex linear-kernel when parameterized with the size of the solution. Finally, we provide polynomial algorithms for \maxtt~and \maxct~when the tournament is sparse, that is when it admits a FAS which is a matching.  
 \end{abstract}

\section{Introduction and related work}
\subsection*{Tournament}
A tournament $\T$ on $n$ vertices is an orientation of the complete graph $K_n$.  
%Thus, given a tournament $\T=(V,A)$, where $V = \{v_i, i\in [n]\}$, %for each $i,j \in [n]$, we have either $v_iv_j \in A$ or $v_jv_i \in %A$.  
A tournament $\T=(V,A)$ can alternatively be defined by an ordering $\V(\T)=(v_1,\dots,v_n)$ of its vertices and a set of \emph{backward arcs} $\bA{A}_{\V}(\T)$
(which will be denoted $\bA{A}(\T)$ as the considered ordering is not
ambiguous), where each arc $a \in \bA{A}(\T)$ is of the form
$v_{i_1}v_{i_2}$ with $i_2 < i_1$.  Indeed, given $\V(\T)$ and
$\bA{A}(\T)$, we define $V = \{v_i, i\in [1,n]\}$ and $A= \bA{A}(\T)
\cup \fA{A}(\T)$ where $\fA{A}(\T) = \{v_{i_1}v_{i_2} : (i_1 <
i_2) \mbox{ and } v_{i_2}v_{i_1} \notin \bA{A}(\T)\}$ is the set of
forward arcs of $\T$ in the given ordering $\V(\T)$. In the following,
$(\V(\T), \bA{A}(\T))$ is called a \emph{linear representation} of the
tournament $\T$.

 A set $A'\subseteq A$ of arcs of $\T$ is a \emph{feedback
  arc set} (or \emph{FAS}) of $\T$ if every directed cycle of $\T$
contains at least one arc of $A'$. It is clear that for any linear
representation $(\V(\T), \bA{A}(\T))$ of $\T$ the set $\bA{A}(\T)$ is
a FAS of $\T$. A tournament is \emph{sparse} if it admits a FAS which
is a matching. 

\subsection*{Considered problems}

We denote by \maxct~the problem of packing the maximum number of arc-disjoint directed cycles in a given tournament. More formally, an input of \maxct~is a tournament $\T$,
an output is a set $\SC=\{c_i, i \in
[|\SC|]\}$, called \emph{cycle packing}, where each $c_i$ is a cycle and for any $i \neq j$ we have
$A(c_i) \cap A(c_j) = \emptyset$ (where $A(c)$ is the set of arcs used in $c$), and the objective is to maximize
$|\SC|$. We denote by \maxtt~the specific case of \maxct~where all cycles $c_i$ are triangles (i.e. directed 3-cycles). In that case, we call the output a \emph{triangle packing}. \\
Moreover, we also consider the parameterized version of the problems, that is \kmaxct{} (resp. \kmaxtt) where, given a tournament $\T$ and a parameter $k$, one has to find if there is a cycle packing (resp. triangle packing) of size at least $k$.

\subsection*{Related work}

%We refer the reader to the appendix of the paper where we precisely recall the definitions of the problems mentioned below as well as the standard definitions about parameterized complexity and approximation.

% We refer the reader to the extended version of the paper~\cite{extBBT18} where we precisely recall the definitions of the problems mentioned bellow as well as the standard definitions about parameterized complexity and approximation.

First it is notable than \maxct~is the dual problem (in the LP sense) of the classical problem {\sc Feedback Arc Set in Tournament} ({\sc FAST}). In this later problem the input is a tournament $\T$ and the output is a FAS $A'$ of $\T$, the objective being to minimize $|A'|$. {\sc FAST} is motivated by numerous practical applications, for instance in voting theory~\cite{borda1781,Condorcet1785},
machine learning~\cite{CohenSS97}, search engine ranking~\cite{DworkKNS01}, and has been intensively studied: \NP-completeness~\cite{Alon06,CTY07}, approximation~\cite{Kenyon-MathieuS07}, \FPT~algorithms~\cite{AlonLS09,Feige09}, kernelization~\cite{BessyFGPPST11}.
However surprisingly, no algorithmic result seems to exist concerning its LP dual, \maxct.\\
On the other hand \maxct~is a particular case of cycle packing in digraphs which is known as an \NP-hard problem (see~\cite{bookJoergen} p.551 for instance) which has been extensively studied too \cite{KrivelevichNY05,Lucchesi78}, since finding arc-disjoint cycles in (di)graphs has many practical applications (for example in biology \cite{Capara03, Dorninger94}).

Alternatively, \maxtt~can be studied as a special case of {\sc 3-set-packing}, by creating the hypergraph on the arc set of the tournament and where each triangle becomes a hyperedge. The {\sc 3-set-packing} problem admits a $\frac{4}{3}+\varepsilon$ approximation~\cite{Cygan13}, implying the same result for \maxtt.
From a structural point of view the problem of partitioning the arc set of a digraph into a collection of triangle has been studied for regular tournaments~\cite{Yuster13}, almost regular tournaments~\cite{AkariaY15} and complete digraphs~\cite{Gardner97}.
%Autres refs/idees de refs:
%\begin{itemize}
%\item  Approx cycle packing (oriented) %\cite{KrivelevichNY05}
%\item packing triangle dans les graphes nono -> ref?
%\item packing de triangles dans des sous classes de %graphes: \begin{itemize}
%\item Packing de triangles dans les tournois reguliers 
%\item packing de triangles dans les digraphes complets (utile?)\cite{Gardner97} (lien direct car introuvable sinon : \url{https://faculty.etsu.edu/gardnerr/pubs/C15.pdf})
%\end{itemize}
%\item Non oriente : graph decomposition \NP complete \cite{DorT97}
%\item parler des sparses: classe specifique de tournois qui...
%\end{itemize}

In the same way, the class of tournaments received a lot of attention in the literature, mainly due to its numerous applications (voting systems \cite{marie1785essai, KarpinskiS10, HemaspaandraSV05}). This gives a new reason to be surprised by this lack of results concerning arc-disjoint cycles in tournaments.\\ 
Finally concerning \FPT~algorithms, few problems are known to admit a $O^*(2^{\sqrt{k}})$ when parameterized by the standard parameter $k$~\cite{pilipczuk2013tournaments}. The parameterized version of {\sc FAST} is one of them~\cite{AlonLS09,Feige09}. To the best of our knowledge, outside bidimensionality theory no packing problems are known to admit such an \FPT~algorithm, and maybe \maxtst~(or even \maxtst) could be a candidate for this and so deserve some attention, especially in the light of the $2^{o(\sqrt{k})}$ lower bound mentioned below.

\subsection*{Our contributions}
We prove in Section~\ref{nphardness} the \NP-hardness of both \maxtt~and \maxct. The reduction also implies \NP-hardness of the "tight" versions of these problems where we look for a packing of size equal to a given feedback arc set, and a $2^{o(\sqrt{k})}$ lower bound (under ETH) for the parameterized versions of these problems. Using the same kind of reduction we also show that deciding if a tournament has a cycle packing and a feedback arc set of the same size is an \NP-complete problem.
It implies in particular that there is no hope for an \FPT~algorithm for {\sc FAST} parameterized above the guaranteed value being the size of a maximal cycle packing~\cite{MahajanRS09}.

Then, from a parameterized point of view, we show in Section~\ref{kernel} that \kmaxtt~is \FPT~as there exists a $O^*(2^k)$ algorithm to solve it. We also show that \kmaxtt~admits a linear-vertex kernel. Finally, we focus on sparse tournaments in Section~\ref{polysparse}.
This class of tournaments is interesting for cycle packing problems. Indeed for instance the problem of packing vertex-disjoint triangles in sparse tournament is \NP~complete in it~\cite{BBT17}. Concerning packing of arc-disjoint cycles, we describe a polynomial algorithm to solve \maxtt~and \maxct~in sparse tournaments.

\section{Notations}
Given a linear representation
$(\V(\T),\bA{A}(\T))$ a triangle $t$ in $\T$ is a triple
$(v_{i_1},v_{i_2},v_{i_3})$ with $i_l < i_{l+1}$ such that either
$v_{i_3}v_{i_1} \in \bA{A}(\T)$, $v_{i_3}v_{i_2} \notin \bA{A}(\T)$
and $v_{i_2}v_{i_1} \notin \bA{A}(\T)$ (in this case we call $t$ a
\emph{triangle with backward arc} $v_{i_3}v_{i_1}$), or
$v_{i_3}v_{i_1} \notin \bA{A}(\T)$, $v_{i_3}v_{i_2} \in \bA{A}(\T)$
and $v_{i_2}v_{i_1} \in \bA{A}(\T)$ (in this case we call $t$ a
\emph{triangle with two backward arcs} $v_{i_3}v_{i_2}$ and
$v_{i_2}v_{i_1}$). We denote by $V(t)=\{v_{i_1},v_{i_2},v_{i_3}\}$ the
vertices of $t$ and by $A(t)$ the set of its arcs. We extend the notation to $V(S)$ and $A(S)$ where $S$ is a set
of triangles.

Given an arc $a=uv$ we define $h(a)=v$ as the head of $a$ and $t(a)=u$ as the tail of $a$. Given two tournaments $\T_1, \T_2$ defined by $\V(\T_l)$ and $\bA{A}(\T_l)$ with $l\in{1,2}$, we denote by $\T=\T_1\T_2$ the tournament called the concatenation of $\T_1$ and $\T_2$, where $\V(\T) = \V(\T_1)\V(\T_2)$ is the concatenation of the two sequences, and $\bA{A}(\T) = \bA{A}(\T_1) \cup \bA{A}(\T_2)$. Given a tournament $\T$ and a subset of vertices $X$, we denote by $\T[X]$ the tournament induced by the vertices of $X$. Moreover, we denote by $\T \setminus X$ the tournament 
$\T[V(\T) \setminus X]$ and say that this tournament is obtained by \emph{removing $X$ from $\T$}.
Given an ordering $\V$ and two vertices $u$ and $v$, we write $u \lo v$ if $u$ is before $v$ in $\V$. Finally, in the following, we will simply write $[n]$ instead of $[1,n]$.

\section{\NP~hardness}
\label{nphardness}
\subsection{Reduction}
We prove the \NP~hardness of \maxtt~using a reduction from \tsatt~(i.e. where each clause has at most 3 literals, and each literal appears at most two times positively and exactly one negative negatively). In the following, denote by $F$ the input formula of an instance of \tsatt. Let $n$ be the number of its variables and $m$ be the number of its clauses. We may suppose that $n\equiv 1, 3 \pmod{6}$ and $m + 1 \equiv 1, 3 \pmod{6}$. Indeed, if its not the case we first add new (unused) variables to get $n \equiv 1, 3\pmod{6}$. Then, as long as $m + 1 \not\equiv 1, 3 \pmod{6}$, we add 6 new variables (to preserve $n \equiv 1, 3
\pmod{6}$)), and one (or two if necessary) clauses using these 6 new
variables. From $F$ we construct a tournament $\T$ which is the concatenation of two tournaments $\Tv$ and $\Tc$ defined below. 

In the following, let $f$ be the reduction that maps an instance $F$ of \tsatt~to a tournament $\T$ we describe now.

\paragraph*{The variable tournament $\Tv$} For each variable $v_i$ of $F$, we define a tournament $V_i$ of order 6 as follows: $\Vv_i(V_i) = (r_i, \bar{x}_i,  x^1_i, s_i, x^2_i , t_i)$ and $\bA{A}_{\Vv}(V_i) = \{s_ir_i, t_ix^1_i\}$. \autoref{fig:gadgetvariable} is a representation of one variable gadget $V_i$. One can notice that the FAS of $V_i$ corresponds exactly to its backward set. We  now define $V(\Tv)$ be the union of the vertex sets of the $V_i$'s and we equip $\Tv$ with the order $\Vv_1\Vv_2\dots \Vv_n$.
Thus, $\Tv$ has $6n$ vertices. We also add the following backward arcs to $\Tv$. Since $n \equiv 1,3 \pmod{6}$, the edges of  the $n$-clique $K_n$ can be packed into (undirected) triangles in polynomial time such that every edge is in exactly one triangle~\cite{SteinerTriples}. Let $\{u_1, \dots, u_n\}$ be an arbitrary enumeration of the vertices of $K_n$. Using a perfect triangle packing $\S_{K_n}$ of $K_n$, we create a tournament $T_{K_n}$ such that $\V'(T_{K_n}) = (u_1, \dots, u_n)$ and $\bA{A}_{\V'}(T_{K_n}) = \{t_ls_l : (s_l, r_l, t_l)$ is a triangle of $\S_{K_n}$ with $s_l <_{\V'} r_l <_{\V'} t_l\}$. Now we set $\bA{A}_{\Vv}(\Tv) = \{uv : u\in V_i, v\in V_j$ for $i\neq j$ and $u_ju_i \in \bA{A}_{\V'}(T_{K_n})\} \cup \bigcup_{i=1}^{n} \bA{A}_{\Vv}(V_i)$. In some way, we ``blew up'' every vertex $u_i$ of $T_{K_n}$ into our variable gadget $V_i$. 
\begin{figure}%
\centering
\includegraphics[width=0.5\columnwidth]{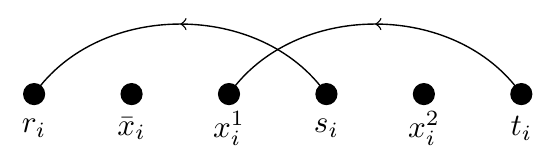}%
\caption{The variable gadget $V_i$. Only backward arcs are depicted, so all the remaining arcs are forward arcs.}%
\label{fig:gadgetvariable}%
\end{figure}

    \paragraph*{The clause tournament $\Tc$} For each of the $m$ clauses $c_j$ of $F$, we define a tournament $C_j$ of order 3 as follows: $\Vc(C_j) = (c^1_j, c^2_j, c^3_j)$ and $\bA{A}_{\Vc}(C_j) = \emptyset$. In addition, we have a $(m+1)^{th}$ tournament denoted by $C_{m+1}$ and defined by $\Vc(C_{m+1}) = (c^1_{m+1}, c^2_{m+1}, c^3_{m+1})$ and $\bA{A}_{\Vc}(C_{m+1}) = \{c^3_{m+1}c^1_{m+1}\}$, that is $C_{m+1}$ is a triangle. We call this triangle the \emph{dummy triangle} , and its vertices the \emph{dummy vertices}. We now define $\Tc$ such that $\Vc(\Tc)$ is the concatenation of each ordering $\Vc(C_j)$ in the natural order, that is $\Vc(\Tc) = (c^1_1, c^2_1, c^3_1,\dots, c^1_m, c^2_m, c^3_m, c^1_{m+1}, c^2_{m+1}, c^3_{m+1})$. So $\Tc$ has $3(m+1)$ vertices. Since $m+1 \equiv 1,3 \pmod{6}$, we use the same trick as above to add arcs to $\bA{A}_{\Vc}(\Tc)$ coming from a perfect packing of undirected triangles of $K_{m+1}$. Once again, we ``blew up'' every vertex $u_j$ of $T_{K_{m+1}}$ into our clause gadget $C_j$. 

\paragraph*{The tournament $\T$}
To define our final tournament $\T$ let us begin with its ordering $\V$ defined by $\V(\T) = \Vv(\Tv)\Vc(\Tc)$. Then we construct $\bA{A}^{vc}(\T)$ the backward arcs between $\Tc$ and $\Tv$. For any $j\in[m]$, if the clause $c_j$ in $F$ has three literals, that is $c_j = l_{1} \vee l_{2} \vee l_{3}$, then we add to $\bA{A}^{vc}(\T)$ the three backward arcs $c^3_j z_{u}$ where $u\in[3]$ and such that $z_{u} = \bar{x}_{i_u}$ when $l_{u}=\bar{v_{i_u}}$, and $z_{u} \in \{x^1_{i_u}, x^2_{i_u}\}$ when $l_{u}=\bar{v_{i_u}}$ in such a way that for any $i \in [n]$, there exists an unique arc $a \in \bA{A}^{vc}(\T)$ with $h(a) = x^1_{i}$. Informally, in the previous definition, if $x^1_{i_u}$ is already "used" by another clause, we chose $z_{u}=x^2_{i_u}$. Such an orientation will always be possible since each variable occurs at most two times positively and once negatively in $F$. If the clause $c_j$ in $F$ has only two literals, that is $c_j = l_{1} \vee l_{2}$, then we add in $\bA{A}^{vc}(\T)$ the two backward arcs $c^2_j z_{u}$ where $u\in[2]$ and such that $z_{u} = \bar{x}_{i_u}$ when $l_{u}=\bar{v_{i_u}}$ and $z_{u} \in \{x^1_{i_u}, x^2_{i_u}\}$ when $l_{u}=v_{i_u}$ in such a way that for any $i \in [n]$, there exists an unique arc $a \in \bA{A}^{vc}(\T)$ with $h(a) = x^1_{i}$.\\
Finally, we add in $\bA{A}^{vc}(\T)$ the backward arcs $c^u_{m+1}\bar{x}_i$ for any $u\in[3]$ and $i\in[n]$. These arcs are called \emph{dummy arcs}. We set $\bA{A}_\V(\T) = \bA{A}_{\Vv}(\Tv) \cup \bA{A}_{\Vc}(\Tc) \cup \bA{A}^{vc}(\T)$. Notice that each $\bar{x}_i$ has exactly four arcs $a \in\bA{A}_\V(\T) $ such that $h(a) =\bar{x}_i$ and $t(a)$ is a vertex of $\Tc$. To finish the construction, notice also that $\T$ has $6n + 3(m+1)$ vertices and can be computed in polynomial time. \autoref{fig:examplereduc} is an example of the tournament obtained from a trivial \tsatt~ instance.
 
 \begin{figure}%
\centering
\includegraphics[width=\columnwidth]{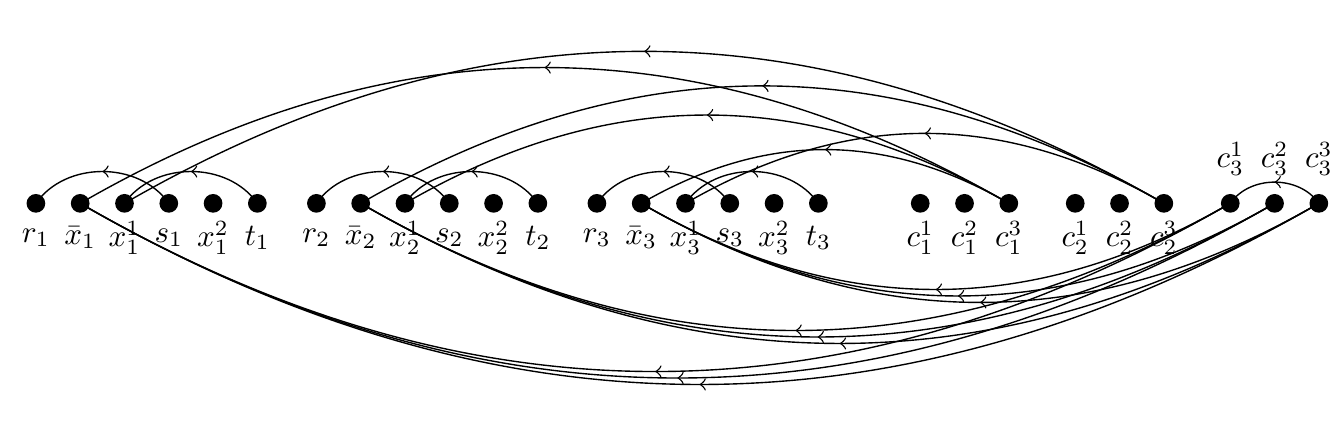}%
\caption{Example of reduction obtained when $F = \{c_1,c_2\}$ where $c_1 = \bar{v}_1 \vee v_2 \vee \bar{v}_3$ and $c_2 = v_1 \vee \bar{v}_2 \vee v_3$. Forward arcs are not depicted. In addition to the depicted backward we have the $36$ backward arcs from $V_3$ to $V_1$, and the $9$ backward arcs from $C_3$ to $C_1$.}%
\label{fig:examplereduc}%
\end{figure}

\subsection{Proof of the reduction}
First of all, observe that in each variable gadget $V_i$, there are only four triangles: let $\t^1_i$, $\t^2_i$, $\t^3_i$ and $\t^4_i$ be the triangles $(r_i,\bar{x}_i, s_i)$, $(r_i, x^1_i,s_i)$, $(x^1_i, s_i, t_i)$ and $(x^1_i, x^2_i, t_i)$, respectively. 
Moreover, notice that there are only three maximal triangle packing of $V_i$ which are $\{\t^1_i, \t^3_i\}$, $\{\t^1_i, \t^4_i\}$ and $\{\t^2_i, \t^4_i\}$. We respectively call these packings  $\S^{\top}_i $, $\S^{\top'}_i$ and $\S^{\bot}_i$.

Given a triangle packing $\S$ of $\T$ and a subset $X$ of vertices, we define for any $x \in X$ the $\S$-local out-degree of the vertex $x$, denoted $d^+_{X\setminus\S}(x)$, as the remaining out-degree of $x$ in $\T[X]$ when we remove the arcs of the triangles of $\S$. More formally, we set: $d^+_{X\setminus \S}(x) = |\{xa : a \in X, xa\in A[X], xa\notin A(\S)\}|$.

\begin{observation}
\label{localoutdegree}
Given a gadget variable $V_i$, we have: 
\begin{enumerate}[(i)]
\item $d^+_{V_i \setminus \S^{\top}_i}(x^1_i) = d^+_{V_i\setminus\S^{\top}_i}(x^2_i) = 1 $ and $d^+_{V_i\setminus\S^{\top}_i}(\bar{x}_i) =3 $,
\item $d^+_{V_i\setminus\S^{\top'}_i}(x^1_i) = 1$, $d^+_{V_i\setminus\S^{\top'}_i}(x^2_i) = 0 $ and $d^+_{V_i\setminus\S^{\top'}_i}(\bar{x}_i) =3 $,
\item $d^+_{V_i\setminus\S^{\bot}_i}(x^1_i) = d^+_{V_i\setminus\S^{\bot}_i}(x^2_i) = 0 $ and  $d^+_{V_i\setminus\S^{\bot}_i}(\bar{x}_i) =4 $,
\item none of $\bar{x}_ix^1_i$, $\bar{x}_is_i$, $\bar{x}_it_i$ belongs to $\S^{\top}_i$ or $\S^{\bot}_i$.
\end{enumerate}
\end{observation}

Informally, we want to set the variable $x_i$ at true (resp. false) when  one of the locally-optimal $\S^{\top'}_i$ or $\S^{\top}_i$ (resp. $\S^{\bot}_i$) is taken in the variable gadget $V_i$ in the global solution.

Now given a triangle packing $\S$ of $\T$, we partition $\S$ into the following sets:
\begin{itemize}
	\item $\SVVV = \{(a,b,c) \in \S : a\in V_i,\ b \in V_j,\ c\in V_k$ with $i<j<k\}$,
	\item $\SVVC = \{(a,b,c) \in \S: a\in V_i,\ b \in V_j,\ c\in C_k$ with $i<j\}$,
	\item $\SVCC = \{(a,b,c) \in \S: a\in V_i,\ b \in C_j,\ c\in C_k$ with $j<k\}$,
	\item $\SCCC = \{(a,b,c) \in \S: a\in C_i,\ b \in C_j,\ c\in C_k$ with $i<j<k\}$,
	\item $\SdVC= \{(a,b,c) \in \S: a,\ b\in V_i,\ c \in C_j\}$,
	\item $\SVdC = \{(a,b,c) \in \S: a\in V_i, b,\ c \in C_j\}$,
	\item $\StV = \{(a,b,c) \in \S: a,\ b,\ c\in V_i\}$,
	\item $\StC = \{(a,b,c) \in \S: a,\ b,\ c\in C_i\}$.
\end{itemize}

Notice that in $\T$, there is no triangle with two vertices in a variable gadget $V_i$ and its third vertex in a variable gadget $V_j$ with $i\neq j$ since all the arcs between two variable gadget are oriented in the same direction. We have the same observation for clauses.\\
In the two next lemmas, we prove some properties concerning the solution $\S$. 

\begin{lemma}
\label{metaclique}
There exists a triangle packing $\Sv$ (resp. $\Sc$) which use exactly the arcs between distinct variable gadgets (resp. clause gadgets). Therefore, we have $|\SVVV| \leq 6n(n-1)$ and $|\SCCC| \leq 3m(m+1)/2$ and these bounds are tight.
\end{lemma}
\begin{proof}
First recall that the tournament $\Tv$ is constructed from a tournament $T_{K_n}$ which admits a perfect packing of $n(n-1)/6$ triangles. Then we replaced each vertex $u_i$ in $T_{K_n}$ by the variable gadget $V_i$ and kept all the arcs between two variable gadget $V_i$ and $V_j$ in the same orientation as between $u_i$ and $u_j$.
Let $u_iu_ju_k$ be a triangle of the perfect packing of $T_{K_n}$. We temporaly relabel the vertices of $V_i$, $V_j$ and $V_k$ respectively by $\{f_i, i\in[6]\}$, $\{g_i, i\in[6]\}$ and $ \{h_i, i\in[6]\}$ and consider the tripartite tournament $K_{6,6,6}$ given by $V(K_{6,6,6}) = \{f_i, g_i, h_i, i\in[6]\}$  and $A(K_{6,6,6}) = \{f_ig_j, g_ih_j, h_if_j : i,j\in[6]\}$.
Then it is easy to check that $\{(f_i, g_j, h_{i+j\pmod{6}}) : i,j \in[6]\}$ is a perfect triangle packing  of $K_{6,6,6}$. 
Since every triangle of $T_{K_n}$ becomes a $K_{6,6,6}$ in $\Tv$, we can find a triangle packing $\Sv$ which use all the arcs between disjoint variable gadgets.

We use the same reasoning to prove that there exists a triangle packing $\Sc$ which use all the arcs available in $\Tc$ between two distinct clause gadget.
\end{proof}

\begin{lemma}
\label{packingsize}
For any triangle packing $\S$ of the tournament $\T$, we have the following inequalities:
\begin{enumerate}[(i)]
	\item $|\SVVV|+|\SCCC| \leq 6n(n-1) + 3m(m+1)/2$
	\item $|\SdVC| + |\SVdC| + |\SVCC| +|\SVVC| \leq \alpha$, where $\alpha = |\bA{A}^{vc}(\T)|$
	\item $|\StV| \leq 2n$
	\item $|\StC| \leq 1$
\end{enumerate}
Therefore in total we have $|\S| \leq  6n(n-1) + 3m(m+1)/2 +  2n + \alpha + 1$.
\end{lemma}
\begin{proof}
Let $\S$ be a triangle packing of $\T$. Remind that we have: $|\S| = |\SVVV|+|\SVVC|+$ $|\SVCC|+|\SCCC|+|\SdVC|+|\SVdC|
+|\StV|+|\StC|$
Inequality $(i)$ comes from Lemma~\ref{metaclique}.
Moreover, we have $|\SdVC| + |\SVdC| + |\SVCC| +|\SVVC| \leq \alpha$ since every triangle of these sets consumes one backward arcs from $\Tc$ to $\Tv$. We have $|\StV| \leq 2n$ since we have at most 2 disjoint triangles in each variable gadget. Finally we also have $|\StC| \leq 1$ since the dummy triangle is the only triangle lying in a clause gadget.
\end{proof}

\begin{theorem}\label{thm:reduc}
$F$ is satisfiable if and only if there exists a triangle packing $\So$ of size $6n(n-1) +  3m(m+1)/2 +  2n + \alpha + 1$ in the tournament $\T$.
\end{theorem}
\begin{proof}
First, let suppose that there exists an assignation $a$ of the variables which satisfies $F$, and let $a^{\top}$ (resp. $a^{\bot}$) be the set of variables set to true (resp. false). 

We construct a triangle packing $\So$ of $\T$ with the desired number of triangles. First, we pick all the disjoint triangles of $\Sv$ and $\Sc$.
By Lemma~\ref{packingsize}, if we also add the dummy triangle $(c^1_{m+1},c^2_{m+1},c^3_{m+1})$ we have $6n(n-1)+3m(m+1)/2 + 1$ triangles in $\So$ until now. 

Then, for any variable $v_i$ of the formula $F$, if $v_i \in a^{\top}$, then we add in $\So$ the triangles $\S^{\top}_i$. Otherwise, we add $\S^{\bot}_i$. One can check that in both cases, these triangles are disjoint to the triangles we just added. Thus, in each $V_i$, we made an locally-optimal solution, so we added $2n$ triangles in $\So$. 

Now we add in $\So$ the triangles $(\bar{x}_i, t_i, c^1_{m+1})$, $(\bar{x}_i, x^1_i, c^2_{m+1})$ and $(\bar{x}_i, x^2_i, c^3_{m+1})$ which will consumes all the dummy arcs of the tournament. 
Recall that in~\autoref{localoutdegree} we mentioned that the vertices $x^1_i$ and $x^2_i$ (resp. $\bar{x}_i$) have an $\S^{\top}_i$-local out-degree both equal to 1 (resp. $\S^{\bot}_i$-local out-degree equals to 4).
Then given a clause $c_j$, let $l$ be one literal which satisfies $c_j$. Assume that the clause is of size 3, since the reasoning is the same for clauses of size 2. If $l$ is a positive literal, say $v_i$, then by denoting $u$ in $\{1,2\}$ the number such that $c^3_jx^u_i$ is a backward arc of $\T$. By the previous recall, we know that there exists $v\in V_i$ such that the arc $x^u_iv$ is available to make the triangle $(x^u_i, v, c^3_j)$. Otherwise, that is if $l$ is a negative literal, say $\bar{v}_i$, then we have $d^+_{V_i\setminus\S^{\bot}_i}(\bar{x}_i) =4$. Three of these four available arcs are used in the triangles which consume the dummy arcs, then we can still make the triangle $(\bar{x}_i,s_i,c^3_j)$. 
Let also $l_{1}$ and $l_{2}$ be the two other literals of $c_j$ (which do not necessarily satisfy $c_j$). Denote by $a_1$ and $a_2$  the vertices of $\Tv$ connected to $c^3_j$ corresponding to the literals $l_{1}$ and $l_{2}$, respectively. Then we add the two following triangles: $(a_1, c^1_j, c^3_j)$ and $(a_2, c^2_j, c^3_j)$.
So we used all the backward arc from $\Tc$ to $\Tv$, and there are no triangles which use two arcs of $\bA{A}_{\Vcv}(\T)$.

Then in the packing $\So$ there are in total $6n(n-1) + 3m(m+1)/2 +  2n + \alpha + 1$ triangles.

Conversely let $\So$ be a triangle packing of $\T$ with $|\So| = 6n(n-1)  +  3m(m+1)/2  +  2n + \alpha + 1$. In the same way as we already did before, we partition $\So$ into the different subsets we defined before. We have $|\So| = |\SoVVV|+|\SoVVC|+|\SoVCC|+|\SoCCC|$ $+|\SodVC|+|\SoVdC|+|\SotV|+|\SotC|$. By Lemma~\ref{packingsize} all the upper bounds described above are tight, that is:
\begin{itemize}
\item $|\SVVV|+|\SCCC| =  6n(n-1)  +  3m(m+1)/2$,
\item $|\SdVC| + |\SVdC| + |\SVCC| +|\SVVC| = \alpha$,
\item $|\StV| = 2n$,
\item $|\StC| = 1$.
\end{itemize}

Let us first prove that $|\SoVVC|+|\SoVCC| = 0$. Let $x = |\SoVVC|+|\SoVCC|$. Since each triangle of the sets $\SoVVC, \SoVCC, \SodVC$ and $\SoVdC$ use exactly one backward arc of $\bA{A}_{\Vcv}(\T)$, it implies that $|\SodVC|+|\SoVdC| \leq \alpha - x$. Moreover, if $x \neq 0$, then we have $|\SoVVV| < |\Sv|$ or $|\SoCCC| < |\Sc|$ because each triangle in $\SoVVC$ (resp. $\SoVCC$) will use one arc between two distinct variable gadgets (resp. clause gadgets) and according to Lemma~\ref{metaclique}, $\Sv$ (resp. $\Sc$) uses all the arcs between distinct variable gadgets (resp. clause gadgets). Finally, we always have $|\SotV| \leq 2n$ and $|\SotC| \leq 1$ by construction. Therefore, if $x \neq 0$, we have $|\So| < |\Sv| + |\Sc| + x + (\alpha - x) + 2n + 1$ that is $|\So| <  6n(n-1)  +  3m(m+1)/2  +  2n + \alpha + 1$, which is impossible. So we must have $x=0$, which implies $\SoVVC = \SoVCC = \emptyset$.\\
Since $|\SotV| = 2n$ and we have at most two arc-disjoint triangles in each variable gadget $V_i$, it implies that $\So[V_i] \in \{\S^{\bot}_i, \S^{\top}_i, \S^{\top'}_i\}$. In the following, we will simply write $\S_i$ instead of $\So[V_i]$. Let consider the following assignation $a$: 
for any variable $v_i$, if $\S_i =\S^{\bot}_i$, then $a(v_i) = false$ and $a(v_i) = true$ otherwise. Let us see that the assignation $a$ satisfies the formula $F$. 
We have just proved that the backward arcs from $\Tc$ to $\Tv$ are all used in $\SdVC$ and $\SVdC$. 
As $|\SotC|=1$ the dummy triangle $C_{m+1}$ belongs to $\So$. So every dummy arc $c_{m+1}^u\bar{x}_i$ is contained in a triangle of $\So$ which uses an arc of $V_i$. Therefore in each $V_i$  we have $d^+_{V_i\setminus\S_i}(\bar{x}_i)\ge 3$.
Moreover, for each clause of size $q$ with $q\in{2,3}$, there are $q$ triangles which use the backward arcs coming from the clause to variable gadgets. Let take $C_j$ one clause gadget of size 3 (we can do the same reasoning if $C_j$ has size 2). By construction the 3 triangles cannot all lie in $\SVdC$. Thus, there is at least one of these triangles which is in $\SdVC$. Let $t$ be one of them, $V_i$ be the variable gadget where $t$ has two out of its three vertices and $\tilde{x}$ be the vertex of $V_i$ which is also the head of the backward arc from $C_j$ to $V_i$. By construction, $\tilde{x}$ corresponds to a literal $l$ in the clause $c_j$. If $l$ is positive, then $\tilde{x} = x^1_i$ or $\tilde{x} = x^2_i$. In both cases, since $t$ has a second vertex in $V_i$, we have $d^+_{V_i\setminus\S_i}(\tilde{x}) > 0$. Thus, using~\autoref{localoutdegree} we cannot have $\S_i =\S^{\bot}_i$ so the assignment sets the positive literal $l$ to $true$, which satisfies $c_j$. Otherwise, $l$ is negative so $\tilde{x}=\bar{x}_i$. Since $\bar{x}_i$ has to use three out-going arcs to consume the dummy arcs and one out-going arc to consume $t$, we have $d^+_{V_i\setminus\S_i}(\bar{x}_i) \geq 4$ and so $\S_i = \S^{\bot}_i$ by~\autoref{localoutdegree}. Therefore, $c_j$ is satisfied in that case too. 

Thus, the assignation $a$ satisfies the whole formula $F$.

\end{proof}

As \tsatt~is \NP-Hard~\cite{papadimitriou2003computational}, this directly implies the following theorem.

\begin{theorem}
\label{triangleNP}
The problem \maxtt~is \NP-Hard.
\end{theorem}

As every cycle in a tournament contains a triangle, 
it is well known that, for vertex disjoint packing, any cycle packing of size $k$ implies a triangle packing of size $k$. This implies that cycle and triangle packing are equivalent in the vertex-disjoint case. However, for the arc-disjoint case, this implication is no longer true in the general case.
Thus, we need to establish the following lemma to transfer the previous \NP-Hardness result to \maxct.

\begin{lemma}
\label{triangleequivcycle}
Given a \tsatt~ instance $F$, and $\T$ the tournament constructed from $F$ with the reduction $f$, we have a triangle packing $\S$ of $\T$ of size $  6n(n-1)  +  3m(m+1)/2  +  2n + \alpha + 1$ if and only if there is a cycle packing $\SC$ of size $  6n(n-1) +  3m(m+1)/2  +  2n + \alpha + 1$.
\end{lemma}

\begin{proof}
Given a cycle packing $\SC$ of $\T$ of size $ 6n(n-1)  +  3m(m+1)/2  +  2n + \alpha + 1$, we partition it into the following sets:
\begin{itemize}
\item $\CV = \{ (v_1, \dots, v_p)\in\SC : \exists i \in [n], \forall k \in [p], v_k \in V_i\}$,
\item $\CC = \{ (v_1, \dots, v_p)\in\SC : \exists j \in [m+1], \forall k \in [p], v_k \in C_j\}$,
\item $\CVV = \{ (v_1, \dots, v_p)\in\SC : \forall k\in [p], \exists i\in[n], v_k \in V_i$ and $(v_1, \dots, v_p) \notin \CV\}$,
\item $\CCC = \{ (v_1, \dots, v_p)\in\SC : \forall k\in [p], \exists j\in[m+1], v_k \in C_j$ and $(v_1, \dots, v_p) \notin \CC\}$,
\item $\CVC = \{(v_1, \dots, v_p)\in\SC :  \exists i \in [n], \exists j \in [m+1], \exists k_1, k_2 \in [p] , v_{k_1} \in V_i, v_{k_2} \in C_j\}$.
\end{itemize}
As we did in the previous proof, we begin by finding upper bounds on each of these sets. First, recall that the FAS of each $V_i$ is 2. Thus, we have $|\CV| \leq 2n$. By construction, we also have $|\CC| \leq 1$.
Secondly, notice that a cycle of $\CVV$ cannot belong to exactly two distinct variable gadgets since the arcs between them are all in the same direction. Thus, the cycles of $\CVV$ have at least three vertices which implies $|\CVV| \leq  6n(n-1) $. We obtain $|\CCC| \leq  3m(m+1)/2 $ using the same reasoning on $\CCC$.
Finally, we have $|\CVC| \leq \alpha$ since each cycle have at least one backward arc.\\
Putting these upper bound together, we obtain that $|\SC| \leq   6n(n-1) +  3m(m+1)/2  +  2n + \alpha + 1$ which implies that the bounds are tight. In particular, cycles of $\CVV$ (resp. $\CCC$) use exactly three arcs that are between distinct variable gadgets (resp. clause gadgets) and all these arcs are used. So we can construct a new cycle packing $\SC'$ where we replace the cycles of $\CVV$ and $\CCC$ by the triangle packings $\Sv$ and $\Sc$ defined in Lemma~\ref{metaclique}. The new solution uses a subset of arcs of $\SC$ and has the same size.\\
The cycles of $\CVC$ use exactly one backward arc of $\bA{A}^{vc}(\T)$ due to the tight upper bound $\alpha$. Moreover, by the previous reasoning, two vertices of a cycle of $\CVC$ cannot belong to two different variable gadgets (resp. clause gadgets). Let $C_j$ be a clause gadget which has three literals (if it has only two literals, the reasoning is analogous). Let $\tilde{x}_{i_k}\in V_{i_k}$ be the head of a backward arc from $c^3_j$ where $k\in[3]$. 
By the previous arguments each arc $c_j^3\tilde{x}_{i_k}$ is contained in a cycle $o_k$ of $\SC$ for $k\in[3]$.
There is at least one $\tilde{x}_{i_k}$ whose next vertex in $o_k$, say $y$, belongs to $V_{i_k}$ since $C_j$ has only two other vertices in addition to $c^3_j$. Without loss of generality, we may assume that $\tilde{x}_{i_3}$ is that vertex. Then, we can replace $o_1$ and $o_2$ by the triangles $(\tilde{x}_{i_1}, c^1_j, c^3_j)$ and $(\tilde{x}_{i_2}, c^2_j, c^3_j)$. The arcs $c^1_j c^3_j$ and $c^2_j c^3_j$ cannot have already been used because $C_j$ is acyclic and we previously consumed all the arcs between clause gadgets. In the same way, we replace the cycle $o_3$ by the triangle $(\tilde{x}_{i_3}, y, c^3_j)$. The arc $yc^3_j$ is available since it could have been used only in the cycle $o_3$.\\ 
We now prove that given a $V_i$, we can restructure every cycle of $\CV[V_i]$ into triangles. Remind that $\CV[V_i]$ have exactly 2 cycles, and notice that by construction one cannot have two cycles each having a size greater than 3. 
First, if the two cycles are triangles, we are done. Then $\CV[V_i]$ contains a triangle, say $\t$, and a cycle, say $o$, of size greater than 3. If $o$ contains the backward arc $s_i r_i$, then by construction $o=(r_i, \bar{x}_i, x^1_i, s_i)$. In that case, we necessary have $\t=(x^1_i, x^2_i, t_i)$ and we can restructure $o$ in the triangle $(r_i, x^1_i, s_i)$. The arc $r_i x^1_i$ is not contained in $\SC$ since the only arcs inside $V_i$ we may have imposed until now are out-going arcs of $x^1_i, x^2_i$ and $\bar{x}_i$. If $o$ contains the backward arc $t_i x^1_i$, then by construction $o=(x^1_i, s_i, x^2_i, t_i)$ and $t=(r_i, \bar{x}_i, s_i)$. In the same way, we can restructure $o$ into $(x^1_i, s_i, t_i)$ whose all the arcs are available.\\
As $\CC$ is already a triangle $\T$ finally has a triangle packing of size $ 6n(n-1)  +  3m(m+1)/2  +  2n + \alpha + 1$. The other direction of the equivalence is straightforward.
 
\end{proof}

The previous lemma and \autoref{triangleNP} directly imply the following theorem:

\begin{theorem}
\label{cycleNP}
The problem \maxct~is \NP-Hard.
\end{theorem}

Let us now define two special cases \maxtttight~(resp. \maxcttight) where, given a tournament $\T$ and a linear ordering $\sigma$ with $k$ backward arcs (where $k$ is the size of an optimal feedback arc set of $\T$), the goal is to decide if there is an arc-disjoint triangle (resp. cycle) packing of size $k$. We call these special cases the "tight" versions of the classical packing problems because as the input admits a FAS of size $k$, any triangle (or cycle) packing has size at most $k$.

We now prove we can construct in polynomial time an ordering of $\T$, the 
tournament of the reduction, with $k$ backward arcs (where $k$ is the threshold value defined in Lemma~\ref{thm:reduc})
 
\begin{lemma}\label{lem:tournamenttight}
Let $\T$ be a tournament constructed by the reduction $f$ of Section~\ref{nphardness}, and $k$ the threshold value defined in Lemma~\ref{thm:reduc}.
We can construct (in polynomial time) an ordering of $\T$ with $k$ backward arcs (implying that $\T$ admits a FAS of size $k$).
\end{lemma}
\begin{proof}
Let us define a linear representation $(\sigma(\T),\bA{A}(\T))$ such that $|\bA{A}(\T)|=k$.

Remember that since $n \equiv 1,3 \pmod{6}$, the edges of  the $n$-clique $K_n$ can be packed into a packing $O$ of $n(n-1)/6$ (undirected) triangles.
Let us first prove that there exists an orientation $\T_{K_n}$ of $K_n$ and a linear ordering $\sigma$ of $\T_{K_n}$ with $|O|$ backward arcs.
Let $\sigma = 1\dots n$.
For each undirected triangle $ijk$ in $O$ where $i < j < k$, we set 
$ki \in \bA{A}(\T_{K_n})$ (implying that $ij$ and $jk$ are forward arcs). As all edges are used in $O$ this defines an orientation for all edges.
Thus, there is only $|O|$ backward arcs in $\sigma$.

Thus, when using the previous orientations $\T_{K_n}$ to construct the variable tournament $T_v$ of the reduction (remember that we blow up each vertex $u_i$ into $6$ vertices $V_i$), we get an ordering with
 $36n(n-1)/6= 6n(n-1)$ backward arcs between two different $V_i$ (more formally, $|\{a \in \bA{A}(\T_v) : \exists i_1 \neq i_2, h(a) \in V_{i_1}, t(a) \in V_{i_2} \}| = 6n(n-1)$).
 
 Following the same construction for the clause tournament $T_c$ we get an ordering with $3m(m+1)/2$ backward arcs between two distinct $C_j$.
 
 Now, as there are two backward arcs in each $V_i$, one backward arc in $C_{m+1}$, and $\alpha$ backward arcs from $T_c$ to $T_v$, the total number of backward arcs is $k$.

\end{proof}

We also prove that $k$ is the size of an optimal feedback arc set of $\T$ too.

\begin{lemma}\label{lem:tightlb}
Let $\T=(V,A)$ be a tournament constructed by the reduction $f$ of Section~\ref{nphardness} and $k$ the threshold value defined in Lemma~\ref{thm:reduc}.
Then, $|FAS^*(\T)| \ge k$, where $FAS^*(\T)$ is an optimal FAS of $\T$.
\end{lemma}
\begin{proof}
We suppose that $\T$ is equiped with the ordering defined in Lemma~\ref{lem:tournamenttight}.
Let $F=FAS^*(\T)$ be an optimal FAS of $\T$.
Given an arc $a$, let $v(a) = \{t(a), h(a)\}$.
Let us partition the arcs of $\T$ into the following sets.
For any $i \in [n], j \in [m+1]$, let us define
\begin{itemize}
\item $A_{V_i} = \{a \in A : v(a) \subseteq V_i \}$
\item $A_{C_j} = \{a \in A:v(a) \subseteq C_j \}$
\item $A_{V_iC_j} = \{a \in A: |v(a) \cap V_i|=|v(a) \cap C_j| = 1 \}$
\item $A_{V_iV_{i'}} = \{a \in A : |v(a) \cap V_i|=|v(a) \cap V_{i'}| = 1 \}$ where $i \neq i'$
\item $A_{C_jC_{j'}} = \{a \in A : |v(a) \cap C_j|=|v(a) \cap C_{j'}| = 1 \}$
where $j \neq j'$
\end{itemize}
For any $i,i' \in [n],\ j,j' \in [m+1]$ and $X \in \{V_i, C_j, V_iC_j, V_iV_{i'}, C_j,C_{j'}\}$,  we also define the corresponding sets $F_{X}$ in $FAS^*(\T)$, where for example $F_{V_i} = F \cap A_{V_i}$.
In addition, for any $j \in [m+1]$ we define $F_{*C_j} = \bigcup_{i \in [n]}F_{V_iC_j}$.

Let $\T'_v$ be the directed graph ($\T'_v$ is not a tournament) obtained by starting from $\T_v$ and only keeping arcs in $A_{V_iV_{i'}}$ for any $i,i' \in [n]$ with $i \neq i'$. 
As $F$ is FAS of $\T$, $F_{VV}=\bigcup_{i,i' \in [n], i \neq i'}F_{V_iV_{i'}}$ must be a FAS of $\T'_v$. As according to Lemma~\ref{metaclique} there is a cycle packing of size $6n(n-1)$ in $\T'_v$, we get $|F_{VV}| \ge 6n(n-1)$.
The same arguments hold for the clause part, and thus with $F_{CC}=\bigcup_{j,j' \in [m+1], j \neq j'}F_{C_jC_{j'}}$, we get $|F_{CC}| \ge 3m(m+1)/2$.
As $C_{m+1}$ is a triangle, we also get $|F_{C_{m+1}}| \ge 1$.

For any $j \in [m]$, let $u_j \in \{2,3\}$ be equal to the size of the clause $j$ (we also have $u_j=|\{a \in \bA{A}(\T) : \exists i \in [n], h(a) \in V_i\mbox{ and }t(a) \in C_j\}|$).
Let $L = \{j \in [m]: |F_{*C_j} \cup F_{C_j}| \ge u_j\}$ be informally the set of clauses where $F$ spends a large (in fact larger than the $u_j$ required) amount of arcs, and $S = [m] \setminus L$.

%Dans ce qui suit : cas où Fcj = 0? utile à préciser? t'as raison mais pas utile
%+ pourquoi a2 et h(a2)t(a1) \notin F?
%car h(a_2) neq t(a), mais j'ajotue %rien dans la preuve
Let us prove that for any $j \in S$, $|F_{C_j}| \ge u_j-1$. 
Let us first consider the case where $u_j=3$.
Suppose by contradiction than $F_{C_j}=\{a\}$ (arguments will also hold for $F_{C_j}=\emptyset$). 
Remember that $\sigma(C_j) = (c^1_j,c^2_j,c^3_j)$ (there are only forward arcs).
As $|F_{*C_j}| \le 1$, there exists $i \in [n]$ and two arcs $a_1, a_2$ not in $F$ such that $t(a_1)=c^3_j$, $h(a_1) \in V_i$, $t(a_2) = h(a_1)$, and $h(a_2) \neq t(a)$. Thus, $(t(a_1),t(a_2),h(a_2))$ is a triangle using no arc of $F$, a contradiction.
As the same kind of arguments holds for the case where $u_j=2$, we get that for any $j \in S$, $|F_{C_j}| \ge u_j-1$ (implying also $|F_{*C_j}|=0$).
 
Let us now prove that $|S| \le 1$. Suppose by contradiction that $|S| \ge 2$.
Let $j_1$ and $j_2$ be in $S$. For any $l \in [2]$, let define $a_l$ such that there exists $i_l \in [n]$ with $t(a_l) \in C_{j_l}$ and $h(a_1) \in V_{i_l}$. Notice that we may have $i_1 = i_2$, but we always have $h(a_1) \neq h(a_2)$. 
Moreover, as $a_i$ is the unique backward arc of $\T$ with $t(a) \in \bigcup_{j \in [m]} C_j$,
%a la place de la phrase précédente, j'aurais mis : as $h(a_l)$ has only one backward arc in $\bigcup_{j \in [m]} C_j$, oui je vois mais faudrait definir qu'un sommet à un backward arc
we get that $a_3=h(a_1)t(a_2)$ and $a_4=h(a_2)t(a_1)$ are forward arcs of $\T$. As $|F_{*C_{j_1}}|=|F_{*C_{j_2}}|=0$ we know that $a_l \notin F$ for $l \in [4]$.
Thus, $(t(a_1),h(a_1),t(a_2),h(a_2),t(a_1))$ is a cycle using no arc of $F$, a contradiction.

Let $L' = \{i \in [n]: \exists a \in \T$ s.t. $ h(a) \in V_i\mbox{ and }t(a) \in C_j,\ j \in S\}$. Notice that if $S = \emptyset$ then $L' = \emptyset$, and otherwise $|L'|=u_{j_0}$, where $S = \{j_0\}$. Let $S' = [n] \setminus L'$.
For any $i \in [n]$, let $\bA{A}_{V_iC_{m+1}} = \bA{A}(\T) \cap A_{V_iC_{m+1}}$.
Recall that $\bA{A}_{V_i C_{m+1}} = c^u_{m+1}\bar{x}_i$ for $u \in [3]$ where $\bar{x}_i \in V_i$. Moreover, for any  $x \in \{\bar{x}_i,x^1_i,x^2_i\}$, let $A_{x V_i}=\{a \in \T : t(a)=x \mbox{ and } h(a) \in V_i\}$. Notice that $|A_{\bar{x}_i V_i}|=4$,  $|A_{x^1_i V_i}|=2$ and $|A_{x^2_i V_i}|=1$.

Let us prove that for any $i \in S'$, $|F_{V_i} \cup F_{V_iC_{m+1}}| \ge 5$. 
If $A_{\bar{x}_i V_i} \subseteq F$, then as $F_{V_i}$ must be a FAS of $V_i$ and 
$A_{\bar{x}_i V_i}$ is not a FAS of $V_i$, there exists at least another arc in $F_{V_i}$ and we get $|F_{V_i}| \ge 5$. Otherwise, $\bA{A}_{V_iC_{m+1}} \subseteq F$ (if it is not the case, there is a cycle $c^u_{m+1}\bar{x}_i v$ where $v \in V_i$ is a out-neighbour of $\bar{x}_i$). Then, as $FAS^*(V_i) \ge 2$, $|F_{V_i} \cup F_{V_iC_{m+1}}| \ge 5$.

% Cas ou L'=0 à préciser? non, devrait marcher direct dans la formule finale
Let us finally prove that for any $i \in L'$, $|F_{V_i} \cup F_{V_i C_{m+1}}| \ge 6$. 
As $i \in L'$, there is an arc $a \in \T$ with $h(a) \in V_i\mbox{ and }t(a) \in C_{j_0}$ where 
$S = \{j_0\}$. Let $x = h(a)$. Notice that $x \in \{\bar{x}_i,x^1_i,x^2_i\}$. As $|F_{*C_{j_0}}|=0$ we get that $A_{xV_i} \subseteq F_{V_i}$ (otherwise there would be a cycle with one vertex in $C_{j_0}$, $x$, and an out-neighbour of $x$ in $V_i$).
\subparagraph*{Case 1: $x=\bar{x}_i$.}As $F_{V_i}$ must be a FAS of $V_i$, $F$ needs two other arcs in $A_{V_i}$ and we get $|F_{V_i}| \ge 6$.
\subparagraph*{Case 2: $x=x^1_i$.}If $A_{\bar{x}_i V_i} \subseteq F$ then 
$|F_{V_i} \cup F_{V_i C_{m+1}}| \ge 6$. Otherwise, as before we get $\bA{A}_{V_i C_{m+1}} \subseteq F$, and as $A_{x^1_i V_i}$ is not a FAS of $V_i$, $F$ need another arc in $V_i$, implying $|F_{V_i} \cup F_{V_i C_{m+1}}| \ge 6$. 
\subparagraph*{Case 3: $x=x^2_i$.}If $A_{\bar{x}_i V_i} \subseteq F$ then 
as  $A_{x^2_i V_i} \cup A_{\bar{x}_i V_i}$ is not a FAS of $V_i$, $F$ need another arc in $V_i$, implying $|F_{V_i}| \ge 6$. Otherwise, as before we get $\bA{A}_{V_i C_{m+1}} \subseteq F$, and as $A_{x^1_i V_i}$ is not a FAS of $V_i$, $F$ need two other arcs in $V_i$, implying $|F_{V_i} \cup F_{V_iC_{m+1}}| \ge 6$.

Putting all the pieces together, we get:

\begin{align*}
|F| = &|F_{VV}| + |F_{CC}| + |F_{C_{m+1}}| + \sum_{j \in L}(|F_{*C_j} \cup F_{C_j}|)+ \sum_{j \in S}(|F_{*C_j} \cup F_{C_j}|)\\
& + \sum_{i \in S'}(|F_{V_i} \cup F_{V_i C_{m+1}}|)+\sum_{i \in L'}(|F_{V_i} \cup F_{V_i C_{m+1}}|) \\
& \ge 
6n(n-1) + \frac{3m(m+1)}{2} + 1 + \sum_{j \in L} u_j + \sum_{j \in S}(u_j-1)+5|S'|+6|L'| \\
& \ge 6n(n-1) + \frac{3m(m+1)}{2} + 1 + \sum_{j \in [m] } u_j + 5n = k
\end{align*}
\end{proof}

Then, using Lemma~\ref{lem:tournamenttight} and Lemma~\ref{lem:tightlb}, we get the \NP-hardness of \maxtttight~and \maxcttight.

\begin{theorem}
\label{tightNP}
The problems \maxtttight~and \maxcttight~are \NP-Hard.
\end{theorem}

Finally, as the size $s$ of the required packing in Theorem~\ref{thm:reduc} verifies $s = \mathcal{O}((n+m)^2)$, 
and as under $ETH$ $3$-SAT cannot be solved in $2^{o(n+m)}$~\cite{cygan2015parameterized}, we also get the following result.

\begin{theorem}
\label{triangleETH}
Under $ETH$, the problems \kmaxtt~and \kmaxct~cannot be solved in $2^{o(\sqrt k)}$.
\end{theorem}

\section{FPT-algorithm and (vertex-)linear kernel for \texorpdfstring{\kmaxtt}{k-MaxTT}}
\label{kernel}
% \begin{qproblem}{\maxtt}
% {$T$ a tournament and $k\in \mathbb{N}$}
% {$true$ if and only if there a collection $\mct$ of arc-disjoint triangles of $T$ such that $|\mct| \geq k$.}
% \end{qproblem}

In this section, we focus on the parameterized version of the \maxtt~problem and provide an \FPT-algorithm for it as well as a kernel with a linear number of vertices.

First using a classical technique of coulor coding~\cite{AlonYZ95} for packing subgraph of bounded size, we obtain the following result.

\begin{theorem}
There exists an algorithm with running time $O^*(2^k)$ to solve \kmaxtt.
\end{theorem}

\begin{proof}
Let $\T$ be an instance of \maxtt~and denote by $n$ its number of vertices, by $m=n(n-1)/2$ its number of arcs. Moreover, we label by $\{e_1,\dots ,e_m\}$ the arcs of $\T$. If $\T$ is a positive instance then it admits a triangle packing $\S$ with $k$ triangles and then containing $3k$ arcs. So we use a a $3k$-perfect family of hash functions from $\{e_1,\dots , e_m\}$
to $\{1, \dots ,3k\}$, that is a set of colorations of the arcs of $\T$  such that for
every subset $U$ of $\{e_1,\dots , e_m\}$ of size $3k$, there exists one
of these colorations that colors the elements of $U$ with $3k$
different colors. Schmidt and Siegal~\cite{SchmidtS90} explicitly provide such a family of colorings of size $2^{O(k)} \log^2 m$ which can be computed in time $O^*(2^{O(k)})$. 

Now for each colouring, we use dynamic programming to obtain a triangle packing of $k$ triangles whose all arcs use different colours. The instance $\T$ is positive if and only if at least one of the colourings contains such a packing. So being given a coloring 
$c:\{e_1,\dots , e_m\}\ra \{1, \dots ,3k\}$, we first compute for every set of 3 colours $\{a,b,c\}$ if the arcs coloured with $a$, $b$ or $c$ induce a triangle using 3 different colours. Then recursively for every set $C$ of $3(p+1)$ colors with $p\in [k-1]$ we test if the arcs colored with the colour of $C$ induce $p+1$ arc-disjoint triangles whose arcs use all the colours of $C$. We do this by looking for every subset $\{a,b,c\}$ of $C$ if there exists a triangle using colors $a$, $b$ and $c$ and a collection of $p$ arc-disjoint triangles whose arcs use all the colors of $C\setminus \{a,b,c\}$. It is clear that $C$ being fixed we can find this collection of triangles in time $O(p^3)=O(k^3)$. Finally, we answer the colour-version of the problem in time $O^*(k^32^k)=O^*(2^k)$. As we have $O^*(2^k)$ different colourings we obtain the announced running time.
\end{proof}

Moreover, we obtain the following kernelization algorithm.

\begin{theorem}
\kmaxtt~admits a kernel with $O(k)$ vertices.
\end{theorem}
\begin{proof}
Let $\x$ be a maximal collection of arc-disjoint triangles of a tournament $\T$ obtained greedily. Let denote by $\Vx$ the vertices of the triangles in $X$ and by $\Ax$ the arcs of the subtournament induced by $\Vx$. Moreover, let $U$ be the remaining vertices of $V(\T)$, that is $U=V(\T)\setminus \Vx$. If $|\x| \geq k$, then $(\T,k)$ is a positive instance of \kmaxtt. Thus we may assume that $|\x| < k$, that is $|\Vx|< 3k$. Moreover, notice that $\T[U]$ is acyclic and that $\T$ does not contain a triangle with one vertex in $V_X$ and two in $U$ (otherwise $\x$ would not be maximal).\\
Let $\B$ be the (undirected) bipartite graph defined by $V(\B) = \Ax \cup \p$ and $E(\B) = \{au : a\in \Ax, u\in \p$ such that $ (h(a),t(a),u)$ forms a triangle in $\T \}$. Let $\M$ be a maximum matching of $\B$ and denote by $\Ap$ (resp. $\Up$) be the vertices of $\Ax$ (resp. $\p$) covered by $\M$. Moreover, we define $\Apb = \Ap \setminus \Ax$ and $\Upb= \Up \setminus \p$.\\
We now prove that $(\Vx \cup \Up, k)$ is a linear kernel of $(\T,k)$. Let $\S$ be an optimal solution of \kmaxtt~ which minimize the number of vertices of $\Upb$ belonging to a triangle of $\S$.
By previous remarks, we can partition $\S$ into $\S_X\cup F$ where $\S_X$ are the triangles of $\S$ included in $\T[\Vx]$ and $F$ are the triangles of $\S$ containing one vertex of $\Vx$ and two vertices of $U$. It is clear that $F$ corresponds to a union of vertex-disjoint stars of $B$ with centres in $U$. Denote by $\Uf$ the vertices of $F$ in $\p$. If $\Uf \subseteq \Up$ then $(\Vx \cup \Up, k)$ is immediately a kernel. Thus, we may suppose that there exists a vertex $x_0$ such that $x_0 \in \Uf \cap \Upb$. ~\\
We will build a tree rooted in $x_0$ with edges alternating between $F$ and $M$. For this let $H_0= \{x_0\}$ and construct recursively the sets $H_{i+1}$ such that
  \[
    H_{i+1}=\left\{
                \begin{array}{ll}
                  N_F(H_{i})\text{ if $i$ is even},\\
                  N_M(H_{i}) \text{ if $i$ is odd}
                \end{array}
              \right.
  \]
Where given a subset $S\subseteq U$, $N_F(S) = \{a \in \Ax : \exists s \in S$ s.t. $(h(a), t(a), s)\in F\}$ and given a subset $S\subseteq \Ax$, $N_M(S) = \{u \in U: \exists a \in \Ax$ s.t. $as\in M\}$. Notice that $H_i\subseteq U$ when $i$ is even and that $H_i\subseteq \Ax$ when $i$ is odd, and that all the $H_i$ are distinct as $F$ is a union of disjoint stars and $M$ 
a matching in $B$. Moreover, for $i\ge 1$ we call $T_i$ the set of edges between $H_i$ and $H_{i-1}$. Now we define the tree $T$ such that $V(T)= \bigcup_i H_i$ and $E(T)= \bigcup_i T_i$. As $T_i$ is a matching (if $i$ is even) or a union of vertex-disjoint stars with centers in $H_{i-1}$ (if $i$ is odd), it is clear that $T$ is a tree.
For $i$ being odd every vertex of $H_i$ is incident to an edge of $M$ otherwise $B$ would contain an augmenting path for $M$, a contradiction. So every leaf of $T$ is in $U$ and incident to an edge of $M$ in $T$ and $T$ contains as many edges of $M$ than edges of $F$.
Now for every arc $a\in \Ax \cap V(T)$ we replace the triangle of $\S$ containing $a$ and corresponding to an edge of $F$ by the triangle $(h(a),t(a),u)$ where $au\in M$ (and $au$ is an edge of $T$).
This operation leads to another collection of arc-disjoint triangles with the same size than $\S$ but it also strictly decreases the number of vertices in $\Upb$ incident to the solution, yielding a contradiction.

	Finally $\Vx\cup \Up$ can be computed in polynomial time and we have $|\Vx \cup \Up| \leq |\Vx| + |M| \leq 2|\Vx| \leq 6k$, which proves that the kernel has $O(k)$ vertices.
\end{proof}

\section{Polynomial algorithm in sparse tournaments}
\label{polysparse}

We now focus on the following optimization problems:
\begin{oproblem}{\maxtst}
{$\T$ a sparse tournament}
{$X$ a collection of arc-disjoint triangles of $\T$}
{Maximize $|X|$}
\end{oproblem}

and

\begin{oproblem}{\maxcst}
{$\T$ a sparse tournament}
{$X$ a collection of arc-disjoint cycles of $\T$}
{Maximize $|X|$}
\end{oproblem}

We will show that these two problems admit a polynomial algorithm.\\ Before this, let $\T$ be a sparse tournament according to the ordering of its vertices $\V(\T)$, that is the set of its backward arcs $\bA{A}(\T)$ is a matching. 
If a backward arc $xy$ of $\T$ lies between two consecutive vertices, then we can exchange the position of $x$ and $y$ in $\V(\T)$ to obtain a sparse tournament with few backward arc. 
So we can assume that the backward arcs of $\T$ do not contain consecutive vertices. 
Moreover, if a vertex $x$ of $\T$ is contained in no backward arc of $\T$ then call $A$ (resp. $B$) the vertices of $\T$ which are before (resp. after) $x$ in $\V(\T)$. 
Let $X_0$ be the set of triangles made from a backward arc from $B$ to $A$ and the vertex $x$. 
As $\T$ is sparse it is clear that $X_0$ is a set of disjoint triangles. Moreover, it can easily be seen that there exists an optimal packing of triangles (resp. cycles) of $\T$ which is the union of an optimal packing of triangles (resp. cycles) of $\T[A]$, one of $\T[B]$ and $X_0$. 
Thus to solve \maxtst~or \maxcst~on $\T$ we can solve the problem on $\T[A]$ and on $\T[B]$ and build the optimal solution for $\T$.
Therefore we can focus on the case where every vertex of $\T$ is the beginning or the end of a backward arc $\bA{A}(\T)$. We will call such a tournament a {\it fully sparse tournament}. So we first address the following optimization problem.

\begin{oproblem}{$\Pi$: {\sc MaxTT for Fully Sparse}}
{$\T$ a fully sparse tournament}
{$X$ a collection of arc-disjoint triangles of $\T$}
{Maximize $|X|$}
\end{oproblem}

Now let order the arcs $e_1, \dots, e_\ba$ of $\bA{A}(\T)$ such that for any $i \in [\ba - 1]$, $h(e_i) \lo h(e_{i+1})$. Moreover, let $\Gf$ be the digraph with vertex set $\Vf = \{e_i : i\in [b]\}$ and arc set $\Af$ defined by: $(e_ie_j)\in \Af$ if $(h(e_i), h(e_j), t(e_i))$ or $(h(e_i), t(e_j), t(e_i))$ is a triangle of $\T$.
%= \{e_i e_j : i,j \in [b] \text{ such that if } i<j, \text{ then } (h(e_i), h(e_j), t(e_i))$ is a triangle of $ \T, (h(e_i), t(e_j), t(e_i))$ is a triangle of $\T$ otherwise$\}$.

Let $\Pip$ be the following problem:
\begin{oproblem}{$\Pip$: {\sc Max Digon-Free Functional Subdigraph }}
{A digraph $\Gf=(\Vf, \Af)$ }
{$X$ a subset of $\Af$ such that the digraph induced by the arcs of $X$ is a functional digraph and digon-free}
{Maximize $|X|$}
\end{oproblem}

Let $X$ be a solution (not necessary optimal) of $\Pip(\Gf)$, and $e_ie_j$ an arc of $X$. We denote by $\Pi(e_ie_j)$ the triangle $(h(e_i), h(e_j), t(e_i))$ if $i<j$ and  otherwise. 
Given a triangle $\Pi(e_ie_j)$, let $s(e_j)$ be the second vertex of $\Pi(e_ie_j)$; in other words, if $\Pi(e_ie_j) = (h(e_i), t(e_j), t(e_i))$, then $s(e_j) = t(e_j)$ and $s(e_j)=h(e_j)$ otherwise.
Informally, $\Pi(e_ie_j)$ corresponds to the triangle formed by the backward arc $e_i$ and one vertex of $e_j$
, that vertex being $s(e_j)$
. In the same way, we define $\Pi(X) = \bigcup_{x\in X} \Pi(x)$.

\begin{claim}
\label{claim:equivPiPi'}
Let $X$ be a solution of $\Pip(\Gf)$. The set $X$ is an optimal solution if and only if $\Pi(X)$ is an optimal solution of $\Pi(\T)$.
\end{claim}

\begin{proof}
Let $e_ie_j$ and $e_ke_l$ be two distinct arcs of $X$. We cannot have $e_i =e_k$ as $X$ induces a functional digraph in $G'$. Without loss of generality, we may assume that $i<k$, that is $h(e_i)\lo h(e_k)$. Moreover, we cannot have $t(e_i)=t(e_k)$ without contradicting that $\T$ is a sparse tournament. 
As $h(e_i)\lo h(e_k)$ the arc $h(e_i)s(e_j)$ is not an arc of $\Pi(e_ke_l)$. Thus if $\Pi(e_ie_j)$ and $\Pi(e_ke_l)$ share a common arc, it means that $s(e_j)t(e_i)=h(e_k)s(e_l)$. But in this case $e_i=e_l$ and $e_j=e_k$, implying $\{e_ie_j,e_ke_l\}$ is a digon of $G'$, which contradict the fact that $X$ is a solution $\Pip(\Gf)$. So, if $X$ is a solution of $\Pip(\Gf)$, then $\Pi(X)$ is an solution of $\Pi(\T)$. Notice that the size of the solution does not change.\\
On the other hand, if $X$ is a subset of the arcs of $G'$ such that $\Pi(X)$ is a solution of $\Pi(\T)$. We cannot have a vertex $e_i$ of $\Gf$ such that $d^+_X(e_i)>1$, since it would imply that the backward arc $e_i$ of $\T$ is covered by at least two triangles of $\Pi(X)$. So $X$ induces a functional subdigraph of $\Gf$. As previously the digraph induced by $X$ is also digon-free otherwise we would have 
two arc-disjoint triangles on only four vertices in $\Pi(X)$, which is impossible. Thus, $X$ is a solution of $\Pip(\Gf)$, and the solution of the same size.\\
The two problems $\Pi$ and $\Pip$ being both maximization problems, they have the same optimal solution.
\end{proof}

Now we show how to solve $\Pi'$ in polynomial time.

\begin{claim}
\label{claim:C3}
If $\Gf$ is strongly connected and has a cycle $C$ of size at least 3 then the solution of $\Pip(\Gf)$ is the number of vertices of $\Gf$.
\end{claim}
\begin{proof}
We construct the arc set $X$ as follows: we start by taking the arcs of $C$. Then, while there is a vertex $x$ which is not covered by any arcs of $X$, we add to $X$ the arcs of the shortest path from $x$ to any vertex of $X$. By construction, every vertex $x$ of every arcs of $X$ verify $d^+_X(x)=1$, and $X$ is digon free. Since $X$ covers every vertex of $\Gf$, $|X|$ is a maximum solution of $\Pip(\Gf)$, that is the number of vertices of $\Gf$.
\end{proof}

A digraph $D$ is a {\it digoned tree} if $D$  arises from a non-trivial tree whose each edge is replaced by a digon.

\begin{claim}
If $\Gf$ is strongly connected and has only cycles of size 2 then $\Gf$ is a digoned tree.
\end{claim}
\begin{proof}
Since $\Gf$ is strongly connected, then for any arc $xy$ of $\Gf$ there exists a path from $y$ to $x$. As $\Gf$ only contains cycles of size 2, the only path from $y$ to $x$ is the directed arc $yx$. So every arc of $\Gf$ is contained in a digon. If $H$ is the underlying graph of $G'$ (without multiple edges) then it is clear that $H$ is a tree otherwise $\Gf$ would contain a cycle of size more than 2.
\end{proof}

\begin{claim}
\label{claim:digontree}
If $\Gf$ is a digoned tree or if $|V(\Gf)|=1$, then the optimal solution of $\Pip(\Gf)$ is $|V(\Gf)|-1$.
\end{claim}
\begin{proof}
The case $|V(\Gf)|=1$ is clear. So assume that $\Gf$ is a digoned tree and let $X$ be a set of arcs of $\Gf$ corresponding to an optimal solution of $\Pip(\Gf)$. 
Then $X$ is acyclic and then has size at most $|V(\Gf)|-1$. Moreover, any in-branching of $\Gf$ provides a solution of size $|V(\Gf)|-1$.
\end{proof}

\begin{lemma}
\label{lemma:polyfunc}
Let $\Gf$ be a digraph with $n$ vertices. Denote by $S_1,\dots, S_p$ terminal strong components of $\Gf$  such that for any $i$ with $1\leq i \leq k$, $S_i$ is a digoned tree or an isolated vertex and for any $i>k$, $S_i$ contains a cycle of length at least 3. Then an optimal solution of  $\Pip(\Gf)$ has
size $n-k$ and we can construct one in polynomial time.
\end{lemma}
\begin{proof}
We can assume that $\Gf$ is connected otherwise we apply the result on every connected component of $\Gf$ and the disjoint union of the solutions produces an optimal solution on the whole digraph $\Gf$.
%
%If $\Gf$ is not connected, denote by $C_1, %\dots, C_p$ its connected components. For any %connected component $C_i$, we can use the %induction and we obtain that $\Pip(C_i) = n_i %- k_i$, where $n_i$ is the number of vertices %of $C_i$, $k_i$ the number of terminal strong %components of $C_i$ which are trees whose %intern edges are digons or isolated vertices. %Moreover, let $X_i$ be the solution of $C_i$. %Then, we can construct a solution $X$, with $X %= \bigcup_{1\leq i\leq p} X_i$. Therefore, we %have $|X| = \sum_{1\leq i \leq p}{n_i - k_i} = %n - k$.
%
So assume that $\Gf$ is connected and let $S$ be a terminal strong component of $\Gf$. If $X$ is an optimal solution of $\Pi'(G')$ then the restriction of $X$ to the arcs of $G'[S]$ is an optimal solution of $\Pi'(G'[S])$. Indeed otherwise we could replace this set of arcs in $X$ by an optimal solution of $\Pi'(G'[S])$ and obtain a better solution for $\Pi'(G')$, a contradiction.
%Suppose there exists a strongly connected %component $S$ of $\Gf$ such that $S$ is %initial but not terminal. Let $\St$ be $\Gf %\setminus S$. We make the induction on $\St$ %and we get a solution $\Xt$ such that $|\Xt| = %\nt - \kt$, where $\nt$ is the order of $\St$ %and $\kt$ the number of terminal strong %components of $\St$ which are trees whose %intern edges are digons or isolated vertices. %Since $\Gf$ is connected and $S$ initial %component, then there exists an arc $xy$ of %$\Gf$ such that $x \in \St$ and $y \in S$. %Moreover, we can make an in-branching in $S$ %rooted in $y$ and let $X_S$ be the set of its %arcs. Thus, we can define $X$ by $X= X_S \cup %\Xt \cup \{xy\}$. Then, we have $|X| = (n - %\nt - 1) + ( \nt - \kt) + 1 = n - \kt = n - %k$.
So by Claim~\ref{claim:C3} and Claim~\ref{claim:digontree} the set $X$ contains at most $\sum_{i=1,\dots ,p} |S_i| - k$ arcs lying in a terminal component of $\Gf$.
Now as every vertex of $\Gf\setminus \bigcup_{i=1,\dots ,p} S_i$ is the beginning of at most one arc of $X$, the set $X$ has size at most $n-k$. Conversely by growing in-branchings
in $\Gf$ from the union of the optimal solutions of $\Pi'(\Gf[S_i])$ for $i=1,\dots ,p$,  by Claim~\ref{claim:C3} and~\ref{claim:digontree} we obtain a solution of $\Pi'(\Gf)$ of size $n-k$ which is then optimal. Moreover, this solution can clearly be built in polynomial time.
%Moreover, we can make an in-branching in $S$ %rooted in $y$ and let $X_S$ be the set of its %arcs. Thus, we can define $X$ by $X= X_S \cup %\Xt \cup \{xy\}$. Then, we have $|X| = (n - %\nt - 1) + ( \nt - \kt) + 1 = n - \kt = n - k$.
%
%Therefore, for any strongly component $S$ of %$\Gf$, $S$ is both initial and terminal, then %$\Gf$ is strongly connected and $k \leq 1$. If %$\Gf$ contains a cycle of size at least $3$, %then $k=0$ and we can now use %the~\autoref{claim:C3} to conclude the proof. %Otherwise, $k=1$ and~\autoref{claim:digontree} %yields the result.
\end{proof}

In all using Claim~\ref{claim:equivPiPi'} and Lemma~\ref{lemma:polyfunc} we can solve \maxtst~in polynomial time.

\begin{lemma}
\label{lemma:maxcst}
In a fully sparse tournament $\T$ the size of a maximum cycle packing is equal to the size of a maximum triangle packing.
\end{lemma}

\begin{proof}
First if $\T$ has an optimal triangle packing of size $|\bA{A}(\T)|$ then as $\bA{A}(\T)$ is a feedback arc set of $\T$, every optimal cycle packing of $\T$ has size $|\bA{A}(\T)|$. 
Otherwise, we build from $\T$ the digraph $G'$ as previously. 
By Lemma~\ref{lemma:polyfunc}, $G'$ has some terminal components $S_1,\dots ,S_k$ which are either a single vertex or induces a digoned tree and every optimal triangle packing of $\T$ has size $|\bA{A}(\T)|-k$.
Let see that no $S_i$ can be a single vertex. Indeed if $S_i=\{e\}$ where $e$ is a backward arc of $\T$, it means that no backward of $\T$ begins or ends between $h(e)$ and $t(e)$ in $\V(\T)$.
As $\T$ is fully sparse, it means that $h(e)$ and $t(e)$ are consecutive in $\V(\T)$ what we forbid previously. Now consider a component $S_i$ which induces a digoned tree in $G'$. Let $\sigma_i$ be the order $\V(\T)$ restricted to the heads and tails
of the arcs of $\T$ corresponding to the vertices of $S_i$. First notice that $\sigma_i$ is an interval of the order $\V(\T)$. Indeed otherwise there exists
two backward arcs $a$ and $b$ of $\T$ such that $a\in S_i$, $b\notin S_i$ and $h(a)$ is before the head or the of $b$ which is before $t(a)$ in $\V(\T)$. But in this case there is an arc in $G'$ from $a$ to $b$ contradicting the fact that $S_i$ is a terminal component of $G'$. 
So we denote $\sigma_i$ by $(x_1,x_2,\dots ,x_l)$ and notice that $x_1$ and $x_2$ are then forced to be the heads of backward arcs belonging to $S_i$. 
If $x_3$ is also the head of backward arc of $S_i$, then we obtain that the three corresponding backward arcs form a 3-cycle in $G'$ contradicting the fact that $S_i$ induces a digoned tree in $G'$. 
Repeating the same argument we show that $l$ is even and that the backward arcs corresponding to the elements of $S_i$ are exactly $x_3x_1$, $x_lx_{l-2}$ and $x_jx_{j-3}$ for all odd $j\in[l]\setminus \{1,3\}$. In other words $S_i$ induces a 'digoned path' in $G'$. Now consider $\S$ an optimal cycle packing of $\T$. Let $X_1$ be the set of backward arcs of $\bA{A}(\T)$ with head strictly before $x_1$ and tail strictly after $x_l$ in $\V(\T)$. And let $\S_1$ be the cycles of $\S$ using at least one arc of $X_1$. It is easy to check that $\S'=(\S \setminus \S_1) \cup \{(h(e),x_1,t(e))\ :\ e\in X_1\}$ is also an optimal cycle packing of $\T$.
Now every cycle of $\S'$ which uses a backward arc of $S_i$ only uses backward arcs of $S_i$ (otherwise it must one arc of $X_1$, which is not possible). Let $\S_i$ be the set of cycles of $\S$ using backward arcs of $S_i$. It is easy to see that $\{x_ix_{i+1}\ :\ i\text{ even and }i\in [l-2]\}$ is a feedback arc set of $\T[\{x_1,\dots ,x_l\}]$ and has size $l/2-1=|S_i|-1$. So we have $|\S_i|\le 
|S_i|-1$.\\ Repeating this argument for $i=1,\dots ,k$ we obtain that $|\S|\le |\bA{A}(\T)|-k$. Thus by Lemma~\ref{lemma:polyfunc} $\S$ has the same size than an optimal triangle packing of $\T$.
\end{proof}

%%
%% Bibliography
%%

\end{document}